\RequirePackage{fix-cm}
\documentclass[11pt,a4paper]{article}
\usepackage[margin=1in]{geometry}
\usepackage[english]{babel}
\usepackage{amsmath,amssymb,amsthm}
\usepackage{graphicx}
\usepackage{microtype}
\usepackage{xcolor}
\usepackage{booktabs}
\usepackage{subcaption}
\usepackage{tikz}
\usetikzlibrary{decorations.pathreplacing}
\usepackage[numbers]{natbib}
\usepackage[colorlinks=true,allcolors=blue]{hyperref}
\newtheorem{theorem}{Theorem}[section]
\newtheorem*{theorem*}{Theorem}
\newtheorem{definition}[theorem]{Definition}
\newtheorem{claim}[theorem]{Claim}

\newtheorem{lemma}[theorem]{Lemma}

\newtheorem{remark}[theorem]{Remark}

\newcommand{\bP}{\mathbf{P}}
\newcommand{\bX}{\mathbf{X}}

\newcommand{\MMS}{\mu}

\title{MMS Allocation for Chores with Item Disposals}

\author{
  Shinuo Ma\thanks{
    Department of Computer Science, the University of Hong Kong.
    \nolinkurl{shinuo_ma@connect.hku.hk}.
  }
  \and
  Xiaowei Wu\thanks{
    IOTSC, University of Macau.
    \nolinkurl{xiaoweiwu@um.edu.mo}.
   } 
}

\date{\today}
\hypersetup{
  pdftitle={MMS Allocation for Chores with Item Disposals},
  pdfauthor={Shinuo Ma and Xiaowei Wu}
}

\begin{document}
\maketitle
\begin{abstract}
Web platforms organize collective work, including content review and
community maintenance, that imposes heterogeneous effort costs on
participants. Fairly distributing these burdens while leaving as few
tasks unassigned as possible is a natural objective for the design
of such platforms.
We study this through the maximin share (MMS) allocation problem with chore
disposals under nonnegative additive costs. Our goal is to minimize
the number of unassigned chores while guaranteeing every agent her
MMS benchmark from the original instance.
We prove that, for every instance with \(n \ge 4\) agents, at most
\(n-3\) chore disposals suffice.
We improve the bound of \(2m/11+n\) disposals established by
Barman et al. (2025), where \(m\) is the number of
chores, and eliminate the dependence on \(m\).
For four agents, we obtain a tight bound: one disposal always suffices,
whereas known instances admit no complete MMS
allocation (Ezra and Garbuz, 2026). We also prove the same result for goods: \(n-3\) duplications suffice to ensure MMS.
\end{abstract}


\section{Introduction}
\label{sec:introduction}

Web platforms in the gig economy allocate indivisible tasks, such as
data annotation and content review, among workers with heterogeneous
effort costs. Fairly distributing these burdens is an important objective
for platforms concerned with worker satisfaction and continued
participation. However, indivisibility can prevent a complete allocation
from meeting every worker's fairness guarantee. When some tasks can be
deferred or omitted, leaving a limited number unassigned provides a way
to preserve fairness. This motivates our study of chore allocation with
disposals: minimizing the number of unassigned tasks while guaranteeing
each worker a fair-share benchmark defined on the original task set.
Such a benchmark must account for both indivisibility and differences
in workers' effort costs.

Fair division provides a formal framework for this task allocation
problem, treating workers as agents and their tasks as indivisible chores. A natural fairness goal is
proportionality: each of the \(n\) agents should incur at most a
\(1/n\) fraction of her total cost for all tasks. However, indivisibility may prevent an agent from partitioning the tasks into
\(n\) equally costly bundles. We therefore consider a benchmark based on
the best partition she can form.
The \emph{maximin share} (MMS) benchmark was introduced by Budish
for goods~\cite{budish2011combinatorial}, extending the divider's
worst-case guarantee in divide-and-choose to \(n\) bundles.
For goods, an agent's MMS is the maximum value she can ensure
for the least valuable bundle by choosing a partition into
\(n\) bundles. For chores, the analogous benchmark reverses
the optimization: her MMS is the minimum cost she can ensure
for the most costly bundle. In both cases, the benchmark
captures her best guarantee when she chooses the partition
but may receive any one of its bundles. An MMS allocation
gives every agent value at least her MMS for goods, or cost
at most her MMS for chores. Although these benchmarks
accommodate indivisibility, MMS allocations need not exist,
even under additive preferences, for either
{goods~\cite{kurokawa2018fair} or chores~\cite{aziz2017chores}.}

We briefly review the extensive literature on approximate
MMS allocations motivated by this nonexistence. An \(\alpha\)-MMS allocation guarantees every agent
value at least \(\alpha\) times her MMS for goods, or cost at most
\(\alpha\) times her MMS for chores.
For additive goods, the guaranteed fraction has increased from
$2/3$~\cite{kurokawa2018fair} to
$3/4$~\cite{ghodsi2018fair},
$10/13$~\cite{heidari2026improved},
and $7/9$~\cite{huang2025fptas}.
For additive chores, the approximation
factor has improved from \(2\)~\cite{aziz2017chores} to
\(4/3\)~\cite{barman2020maximin}, {\(11/9\)~\cite{huang2021chores},
and \(13/11\)~\cite{huang2023scheduling}.}
This relaxes the fairness guarantee while requiring all items to be allocated.

An alternative is to adjust the supply of items while keeping every
agent's original MMS benchmark. Barman
et al.~\cite{barman2025adjusted} formalize this approach through
duplication of goods and disposal of chores. For monotone chore costs,
they show that discarding at most \(m/e\) of the \(m\) chores suffices,
and that this bound is asymptotically tight. For additive costs, they
obtain a bound of \(2m/11+n\) disposals. On the goods side, Akrami
et al.~\cite{akrami2026augmentation} show that \(n-2\) distinct extra
copies suffice for additive valuations. These results preserve exact
fairness by modifying which items must be allocated, rather than
changing the shares against which {fairness is measured.}

We study this approach for indivisible chores with nonnegative
additive costs. A disposal leaves a chore unassigned; all MMS
thresholds are computed before any chores are removed. Our objective
is to minimize the number of disposals, with particular attention to
bounds depending on the number of {agents. We ask:}
\begin{quote}
\emph{How many chores must be left unassigned to guarantee every agent
her {original maximin share?}}
\end{quote}

\subsection{Our Results}

We prove that exact MMS fairness for additive chores can be
guaranteed with at most $n-3$ disposals. We improve the $2m/11+n$ disposal bound of Barman et al.~\cite{barman2025adjusted} and eliminate
the dependence on the total number {$m$ of chores.}

\begin{theorem*}[Main result]
For every chore instance with $n\ge4$ agents and nonnegative
additive costs, there exists an allocation that leaves at most
$n-3$ chores unassigned and gives every agent her {original
MMS guarantee.}
\end{theorem*}

For four agents, we therefore obtain an optimal bound
of one disposal while preserving every agent's original MMS
guarantee. With this result, we complement a line of work on MMS approximations
specifically for four agents. For additive goods, the guarantees
have improved from $3/4$~\cite{kurokawa2018fair} to
$4/5$~\cite{ghodsi2018fair}. For additive chores,
Qiu et al.~\cite{qiu2026allbutone} establish an all-but-one
guarantee: three agents receive their MMS guarantees, while
the fourth incurs cost at most $4/3$ times her MMS. For the lower bound, we use the result of Ezra and
Garbuz~\cite{ezra2026impossibility}, who construct a four-agent,
eleven-chore instance admitting no \(\alpha\)-MMS {allocation
for \(\alpha<31/30\).}

In our proof, we combine a matching-based framework with round-robin
guarantees~\cite{zhou2022efx,aziz2019goodschores} and a structural
argument for repartitioning chores. With matching and round-robin, we can already obtain a bound of \(n-2\) disposals; with our structural
argument, we save one additional disposal, establishing the
\(n-3\) bound.
Matching-based methods have a substantial history in fair
allocation. Gan et al.~\cite{gan2019envy} use matchings to find
envy-free house assignments whenever they exist.
Kobayashi et al.~\cite{kobayashi2023efx} use related techniques
to establish EFX allocations for several classes of chore
instances, and Qiu et al.~\cite{qiu2026allbutone} use them to
guarantee MMS for all but one agent when allocating chores.
For goods, Bu et al.~\cite{bu2026comparison} use matchings to
compute fair allocations {through comparison-based queries.}

\subsection{Further Related Work}

We next review work on exact MMS existence under restrictions on both
the number of items and agents' preferences. For additive goods,
Feige et al.~\cite{feige2022negative} prove existence whenever
\(m\le n+5\) and give a three-agent, nine-good counterexample.
Hummel~\cite{hummel2023lower} extends the existence guarantee to
\(m\le n+6\) when \(n\ne3\). For additive chores,
Hsu~\cite{hsu2024mixedmanna} establishes existence whenever
\(m\le n+5\). Hummel also shows that, for both goods and chores
and every fixed positive integer \(c\), MMS allocations exist
whenever \(m\le n+c\) and \(n\) {is sufficiently large.} More recently, Feige~\cite{feige2026fewitems} proves that the maximum
number of items for which MMS existence is universally guaranteed is
\(n+\Theta(\log n)\), both for goods and for chores under general
monotone preferences.

Another line of work identifies classes of preferences that admit
exact MMS allocations without restricting the number of items.
Ebadian et al.~\cite{ebadian2022easy} establish existence for both
goods and chores under weakly lexicographic preferences and under
personalized bivalued preferences with an integer ratio between
each agent's two positive values or costs.
Garg et al.~\cite{garg2025binpacking} establish existence for
factored chore costs. In a recent preprint,
Zhang~\cite{zhang2026personalized} removes the integer-ratio
restriction for personalized bivalued goods and chores, allowing
each agent to have her own pair of nonnegative values or costs.
We also know that exact MMS allocations exist in some nonadditive settings:
Barman and Verma~\cite{barman2021matroid} establish it for goods
with matroid-rank valuations, while Barman
et al.~\cite{barman2023supermodular} establish it for chores with
supermodular costs whose marginal costs {belong to \(\{0,1\}\).}

Other work relaxes the MMS requirement by changing the benchmark
or the agents receiving it.
Ordinal approximations compare an allocation with the MMS obtained
by partitioning into a different number of
bundles~\cite{akrami2024improving,hosseini2022ordinal,garg2025binpacking}.
For goods, Hosseini and Searns~\cite{hosseini2021guaranteeing} study
guarantees that satisfy a fraction of the agents. For chores, Qiu
et al.~\cite{qiu2026allbutone} preserve exact MMS for \(n-1\) agents
while bounding the remaining agent's cost by an approximation factor.

Leaving items unassigned has also been studied under envy-based
fairness notions, with guarantees on either the welfare retained
or the number of unassigned items. For additive goods,
Caragiannis et al.~\cite{caragiannis2019donating} show that a
partial allocation that is envy-free up to any good (EFX) can
retain at least half the optimal Nash social welfare of the
original instance. Chaudhury
et al.~\cite{chaudhury2021charity} obtain EFX allocations under
monotone valuations while leaving fewer than \(n\) goods unassigned.
Berger et al.~\cite{berger2022almost} show that at most one
unassigned good suffices for four agents with additive valuations,
and Mahara~\cite{mahara2024extension} establishes a bound of
\(n-2\) unassigned goods for general monotone valuations.
For chores, Tao et al.~\cite{tao2025binary} obtain partial
envy-free allocations with at most \(n-1\) unassigned chores
when all marginal costs {belong to \(\{0,1\}\).}

\section{Preliminaries}
\label{sec:preliminary}

We study the allocation of a set $M$ of $m$ indivisible chores among a set $N=[n]$ of agents. Each agent $i\in N$ has an additive cost function $c_i:2^M\to\mathbb{R}_{\geq 0}$. We denote an instance by $I=(N,M,\{c_i\}_{i\in N})$.
For any positive integer $k$, a $k$-partition of $M$ is a tuple $\mathbf{P}=(P_1,\ldots,P_k)$ of pairwise disjoint, possibly empty bundles whose union is $M$. We denote the set of all such partitions by $\prod_k(M)$. An allocation is an $n$-partition $\bX=(X_1,\ldots,X_n)$ of a subset of chores in $M$, where $X_i$ is the bundle assigned to agent $i$.
An allocation is \emph{complete} if $\bigcup_{i\in N}X_i=M$.
For any integer $t$, we use $[t]$ to denote the set $\{1,2,\ldots,t\}$.
For ease of notation, for a bundle $S$ and an item $e$, we use $S+e$ and $S-e$ to denote $S\cup \{e\}$ and $S\setminus \{e\}$, respectively. 
We focus on the fairness notion of maximin share, formally {defined as follows.}
\begin{definition}[MMS]
    Given an instance $I=(N,M,\{c_i\}_{i\in N})$, the maximin share (MMS) value of agent $i$ is defined as
    \begin{equation*}
        \MMS_i :=\min_{\bX \in \prod_n(M)} \max_{j \in [n]}  \left\{c_i(X_j)\right\}.
    \end{equation*}
    We call an $n$-partition $\mathbf{P} = (P_1,P_2,\ldots,P_n)$ of $M$ an {MMS partition} of agent $i$ if $c_i(P_j)\leq \MMS_i$ holds {for all $j\in [n]$.}
\end{definition}

If $\MMS_i=0$ for some agent $i$, then all her chore costs are zero, so we obtain a complete MMS allocation by assigning all chores to her. We therefore assume that $\MMS_i>0$ for all $i\in N$ and scale the cost functions so that $\MMS_i = 1$ for all $i\in N$.
We seek an allocation, not necessarily complete, such that $c_i(X_i) \leq 1$ for all agents.
Since complete MMS allocations are not guaranteed to exist, our goal is to minimize the number of unallocated items, i.e., $|M\setminus (\bigcup_{i\in N}X_i)|$, while preserving each agent's MMS guarantee.
For every $n\geq 3$, some instances require at least {one unallocated item~\cite{ezra2026impossibility,li2025settling}.}

\begin{definition}[MMS Allocation]
    An allocation $\bX = (X_1,\ldots,X_n)$ is MMS if $c_i(X_i)\leq 1$ for all $i\in N$. It is a complete MMS {allocation if $\bigcup_{i\in N}X_i=M$.}
\end{definition}

We also consider {proportionality and PROP1.}

\begin{definition}[PROP Allocation]
    An allocation $\bX=(X_1,\ldots,X_n)$ is \emph{proportional} (PROP)
    if $c_i(X_i)\leq {c_i(M)}/{n}$ holds {for all $i\in N$.}
\end{definition}

\begin{definition}[PROP1 Allocation]
    An allocation $\bX=(X_1,\ldots,X_n)$ is
    \emph{proportional up to one chore} (PROP1)
    if, for every agent $i\in N$
    with $X_i\neq\emptyset$, there exists a chore $e\in X_i$ such that
    \[
        c_i(X_i - e)\leq {c_i(M)}/{n}.
    \]
\end{definition}

We will describe partial MMS allocations by explicitly stating how many
chores remain unallocated. Throughout, we compute the MMS benchmark from
the original set $M$, even when the final allocation leaves {some chores
unassigned.}

\begin{definition}[MMS-feasibility Graph]\label{chores-MMS-feasibility graph}
    Given an $n$-partition $\mathbf{P}=(P_1,\ldots,P_n)$,
    its {MMS-feasibility graph} is the bipartite graph $G = (N \cup \mathbf{P}, E)$
    with agent vertices $i\in N$, bundle vertices $P_j$ indexed by
    $j\in N$, and edge set
    $
        E=\{(i,P_j): i,j\in N,\ c_i(P_j)\leq 1\}.
    $
    For any set of agents $S\subseteq N$, we denote {its neighborhood by}
    $
        L(S)=\{P_j: j\in N,\ c_i(P_j)\leq 1
        \text{ for some } i\in S\}.
    $
\end{definition}

We use the MMS-feasibility graph throughout our analysis.
For example, consider the case of three agents, and form the MMS-feasibility graph of agent $3$'s MMS partition.
If the graph admits a perfect matching, then we obtain an MMS allocation. 
Otherwise, since every agent has at least one neighbor (e.g., her favorite bundle) and agent $3$ accepts all three bundles, some bundle $P_i$ is adjacent only to agent $3$.
After assigning $P_i$ to agent $3$, we obtain a remaining instance with total cost below $2$ for each remaining agent. We then use a divide-and-choose protocol~\cite{qiu2026allbutone} to guarantee both agents
cost at most $1$ after discarding at most one chore from {the divider's bundle.}

\section{Allocation Framework}
\label{four agents}

In this section, we present our main theorem, beginning with special cases to explain {the proof structure.}

\paragraph{Four Agents.}
Consider agent~\(4\)'s MMS partition, and suppose that the other three
agents collectively accept at most two bundles. Assigning a bundle
\(P_j\) that all three reject to agent~\(4\) leaves
\[
c_i(M\setminus P_j)<3
\qquad\text{for }i\in\{1,2,3\}.
\]
Round-robin then guarantees at most two disposals. To obtain one disposal, we seek a new partition into three bundles that are MMS-feasible to some remaining agent. However, we cannot guarantee such a partition from the total-cost bound alone.
We show that, for a suitable choice of \(P_j\), such a partition exists after temporarily setting aside at most one chore. In the new feasibility graph, one agent accepts all three bundles, and each other agent accepts at least one. This restores the starting structure of the three-agent case. If the graph admits a perfect matching, we allocate the matched bundles and discard the set-aside chore. Otherwise, we repeat the matching step from the three-agent case and assign a bundle that is MMS-feasible to only one agent. The two remaining agents then use divide-and-choose to assign all unallocated chores, including any chore temporarily set aside. Thus, at most one
chore is discarded {in either case.}

\medskip

We now formalize the matching-based technique illustrated above and use
it to reduce the proof of our main result to a {single remaining case.}

\begin{theorem}\label{thm:n-3}
Given any instance $I=(N,M,\{c_i\}_{i\in N})$ with $n\ge4$,
there exists an MMS allocation with at most {$n-3$ unallocated chores.}
\end{theorem}

Recall that we assume $\MMS_i=1$ for every agent. 
Therefore, we have $c_i(M)\le n$ and $c_i(e)\le1$ for every chore $e$.
Furthermore, each agent has an MMS partition into $n$ bundles of cost at most $1$.
Let $\bP=\{P_1,\ldots,P_n\}$ be an MMS partition of agent~$n$,
and construct its MMS-feasibility graph. Agent~$n$ accepts all
$n$ bundles. Every other agent also has a neighbor.
If the graph admits a perfect matching, then we obtain a complete MMS allocation. Otherwise, by Hall's theorem, we find a set of agents $S$ with
$|L(S)|<|S|$; choose such a set of maximum cardinality.
Note that $2\le |S|\le n-1$ and agent $n$ is not contained in $S$.
By the maximality of $S$, we can match every agent outside $S$
to a bundle rejected by all agents in $S$. We formalize this observation in the following lemma.

\begin{lemma}\label{lemma:pm_outside_S*}
There exists a matching saturating $N\setminus S$ that uses
only {bundles in $\bP\setminus L(S)$.}
\end{lemma}
\begin{proof}
Restrict the graph to the agents in $N\setminus S$ and bundles
in $\bP\setminus L(S)$. If no matching saturates $N\setminus S$,
we find, by Hall's theorem, a nonempty set $S'\subseteq N\setminus S$
with fewer than $|S'|$ neighbors in this restricted graph. Its
neighborhood is $L(S')\setminus L(S)$, and hence
\[
|L(S')\setminus L(S)|<|S'|.
\]
The additional neighbors contributed by $S'$ are exactly those
outside $L(S)$. Since $S$ and $S'$ are disjoint,
\[
\begin{aligned}
|L(S\cup S')|
&=|L(S)|+|L(S')\setminus L(S)|\\
&<|S|+|S'|=|S\cup S'|.
\end{aligned}
\]
Thus $S\cup S'$ is a Hall violator strictly larger than $S$,
contradicting the choice of $S$. Figure~\ref{fig:overview-hall-reduction} illustrates {this matching technique.}
\end{proof}

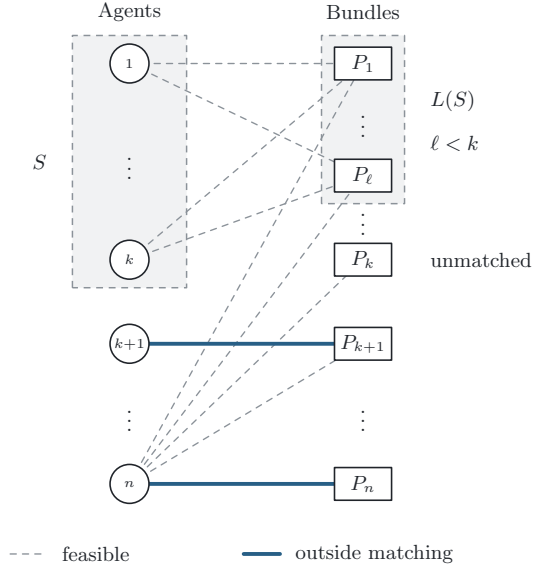
\begin{figure}[htbp]
  \centering
  \resizebox{0.50\linewidth}{!}{%
  \begin{tikzpicture}[
    x=.08cm,y=.072cm,line join=round,
    every node/.style={
      anchor=center,text=mmsInk,
      font=\fontsize{9}{10.03}\selectfont,
      inner sep=0pt,outer sep=0pt
    },
    feasible/.style={
      draw=mmsGray,line width=.6pt,line cap=round,
      dash pattern=on 2.7pt off 2.7pt
    },
    assigned/.style={
      draw=mmsBlue,line width=1.7pt,line cap=round
    }
  ]
    \definecolor{mmsInk}{RGB}{32,36,42}
    \definecolor{mmsBlue}{RGB}{40,95,134}
    \definecolor{mmsGray}{RGB}{139,146,153}
    \definecolor{mmsLight}{RGB}{241,242,243}

    \path[use as bounding box] (0,-13) rectangle (116,125);

    \path[
      fill=mmsLight,draw=mmsGray,line width=.65pt,
      dash pattern=on 3pt off 2.5pt
    ] (18,51) rectangle (40,105);

    \path[
      fill=mmsLight,draw=mmsGray,line width=.65pt,
      dash pattern=on 3pt off 2.5pt
    ] (66,69) rectangle (82,105);

    \foreach \y in {99,57} {
      \foreach \z in {99,75} {
        \draw[feasible] (29,\y) -- (74,\z);
      }
    }

    \foreach \z in {99,75,57,39} {
      \draw[feasible] (29,9) -- (74,\z);
    }

    \draw[assigned] (29,39) -- (74,39);
    \draw[assigned] (29,9) -- (74,9);

    \foreach \a/\y in {1/99,k/57,{k+1}/39,n/9} {
      \path[fill=white,draw=mmsInk,line width=.7pt]
        (29,\y) circle[radius=.3cm];
      \node at (29,\y) {$\scriptstyle\a$};
    }

    \foreach \b/\y in {1/99,\ell/75,k/57,{k+1}/39,n/9} {
      \path[fill=white,draw=mmsInk,line width=.7pt]
        (68.5,{\y-3.5}) rectangle (79.5,{\y+3.5});
      \node at (74,\y) {$P_{\b}$};
    }

    \foreach \x/\y in {29/78,29/24,74/87,74/66,74/24} {
      \node at (\x,\y) {$\vdots$};
    }

    \node[font=\fontsize{9}{9.44}\selectfont]
      at (29,110) {Agents};
    \node[font=\fontsize{9}{9.44}\selectfont]
      at (74,110) {Bundles};

    \node[anchor=east] at (13,78) {$S$};
    \node[anchor=west] at (87,91) {$L(S)$};

    \node[
      anchor=west,
      font=\fontsize{9}{8.26}\selectfont
    ] at (87,82) {$\ell<k$};

    \node[
      anchor=west,
      font=\fontsize{9}{8.26}\selectfont
    ] at (87,57) {unmatched};

    \draw[feasible] (6,-6) -- (13,-6);
    \draw[assigned] (51,-6) -- (58,-6);

    \node[
      anchor=west,
      font=\fontsize{9}{8.26}\selectfont
    ] at (16,-6) {feasible};

    \node[
      anchor=west,
      font=\fontsize{9}{8.26}\selectfont
    ] at (61,-6) {outside matching};
  \end{tikzpicture}%
  }
  \caption{Matching-Based Technique}
  \label{fig:overview-hall-reduction}
\end{figure}

After matching the agents outside \(S\), we allocate the remaining
chores among the agents in \(S\) by round-robin, using the following
guarantee.

\begin{lemma}[{\cite[Lemma~13]{zhou2022efx};
\cite[Proposition~2]{aziz2019goodschores}}]
\label{lem:prop}
For every additive chore instance, consider the round-robin algorithm:
agents take turns cyclically in a fixed order, each selecting a remaining
chore of minimum cost to her, with ties broken arbitrarily, until all
chores have been allocated. The resulting allocation is PROP1, and at
least one agent receives a PROP bundle.
\end{lemma}

Here, we measure proportionality with respect to the chores and
agents participating in round-robin. We will show that these
proportional thresholds are below the original MMS benchmark \(1\).
The PROP bundle therefore requires no disposal, while each other
bundle requires at most one. We therefore discard at {most \(|S|-1\) chores.}
We combine this bound with the matching outside \(S\) in the
following lemma to settle {the case \(|S|\le n-2\).}

\begin{lemma}\label{lem:non-extreme-hall}
If $|S|\le n-2$, then there exists an
MMS allocation with at most {$n-3$ unallocated chores.}
\end{lemma}

\begin{proof}
By Lemma~\ref{lemma:pm_outside_S*}, assign each agent in
$N\setminus S$ a distinct MMS-feasible bundle outside $L(S)$.
There are $n-|S|$ such bundles. Since none belongs to $L(S)$,
every agent in $S$ assigns cost greater than $1$ to each of them.
Writing $R$ for the remaining chores, we have, by additivity,
\[
c_i(R)<c_i(M)-(n-|S|)
\le n-(n-|S|)=|S|
\qquad(i\in S).
\]
Apply Lemma~\ref{lem:prop} to $S$ and $R$. Since
$c_i(R)/|S|<1$ for every $i\in S$, the PROP bundle is
MMS-feasible. For each other nonempty bundle, we can, by PROP1, remove
a chore to reduce its owner's cost to at most her
proportional threshold, and hence below $1$. Empty bundles
{are already MMS-feasible.}

At most $|S|-1\le n-3$ chores are discarded. The bundles assigned
by matching are disjoint from $R$, and round-robin initially
allocates every chore in $R$. Thus these are the only disposals
in the complete construction.
\end{proof}

\begin{remark}
Using the same argument, we also obtain an MMS allocation with at most
\(n-2\) disposals for every additive chore instance with $n\ge2$ agents. We thereby improve the previous bound of \(\frac{2m}{11}+n\) disposals for additive
chores \cite{barman2025adjusted} and match the number of duplications sufficient for
additive goods \cite{akrami2026augmentation}.
\end{remark}

\medskip

We now consider \(S=N\setminus\{n\}\). Let
\[
J=\{j\in[n]:P_j\notin L(S)\}
\]
be the indices of the bundles rejected by every agent in \(S\). Since
\(|L(S)|<n-1\), we have \(|J|\ge2\). Agent \(n\) accepts each
of these bundles, whereas $c_i(P_j)>1$ holds for every $i\in S$ and $j\in J$.
Thus, assigning any such \(P_j\) to agent \(n\) leaves
\[
c_i(M\setminus P_j)=c_i(M)-c_i(P_j)<n-1
\qquad \forall i\in S.
\]

By this bound, we know that every remaining agent accepts at least one bundle in any partition of the remaining chores into \(n-1\) bundles. We seek a partition such that one remaining agent accepts every bundle, allowing us to repeat the matching argument.
A perfect matching in the new feasibility graph completes the allocation. Otherwise, every Hall violator excludes the agent
who accepts all bundles and therefore has size at most \(n-2\).
Matching outside a maximum-cardinality Hall violator and applying
round-robin then yields at most \(n-3\) disposals.

Since we cannot guarantee such a partition from the total-cost bound alone, we allow at most \(n-3\)
chores to be temporarily set aside when constructing the partition.
For \(n=4\), this recovers the starting structure of the three-agent
case with at most one chore set aside. We use the reduction below to ensure
that these temporary removals do not increase the final disposal
count. We formalize the required packing property with respect
to a single agent's cost function using the following definition.

\begin{definition}[$(k,l)$-Packable]
    A chore set $M'$ is \emph{$k$-packable} to agent $i$ if it can be partitioned into $k$ bundles, each of cost at most $1$ to that agent.
    A chore set $M'$ is \emph{$(k,l)$-packable} to agent $i$ if removing at most $l$ chores leaves a \emph{$k$-packable} set {for that agent.}
\end{definition}

In the following lemma, we show that it suffices to find \(j\in J\)
such that \(M\setminus P_j\) is \emph{\((n-1,n-3)\)-packable} to
some agent in \(S\). We use such a packing to construct
a new feasibility graph. If this graph admits a perfect matching,
we allocate the matched bundles and discard only the at most
\(n-3\) set-aside chores. Otherwise, we match outside a
maximum-cardinality Hall violator and return all set-aside chores
before running round-robin among the unmatched agents. In this
case, only the at most \(n-3\) chores removed from the round-robin
allocation are discarded, so the temporary removals contribute
{no additional disposals.}

\begin{lemma}\label{lem:near-packing-reduction}
Suppose there exist $i\in S$ and $j\in J$ such that
$M\setminus P_j$ is $(n-1,n-3)$-packable to agent $i$.
Then we can find an MMS allocation with at most {$n-3$ unallocated
chores.}
\end{lemma}

\begin{proof}
Assume without loss of generality that $i=1$. By the assumed
packability, there is a set $D\subseteq M\setminus P_j$ with $|D|\le n-3$,
and a partition $\{B_1,\ldots,B_{n-1}\}$ of
$(M\setminus P_j)\setminus D$ satisfying
\[
c_1(B_h)\le1\qquad \forall h\in[n-1].
\]
Temporarily set aside $D$ and allocate $P_j$ to agent~$n$.
This allocation is MMS-feasible for agent~$n$, and every remaining
agent satisfies
\[
c_i(M\setminus P_j)<n-1\qquad \forall i\in S.
\]
We include the chores in $D$ in this bound, so we can use it
again when they are returned.

Construct the MMS-feasibility graph between the $n-1$ agents in
$S$ and the bundles $B_1,\ldots,B_{n-1}$. Agent~$1$ is adjacent
to every bundle by construction. Every other agent has a neighbor
because, by the nonnegativity of costs, we have
\[
\sum_{h=1}^{n-1}c_i(B_h)
=c_i\bigl((M\setminus P_j)\setminus D\bigr)
\le c_i(M\setminus P_j)<n-1.
\]
In particular, these $n-1$ bundles cannot all have cost greater
than $1$ to the same agent.

If the new graph has a perfect matching, assign its matched
bundles to the agents in $S$. Every agent then receives an
MMS-feasible bundle, and precisely the chores in $D$ remain
unallocated. This uses at {most $|D|\le n-3$ disposals.}

Otherwise, let $T$ be a maximum-size Hall violator in the new
graph. It excludes agent~$1$, who accepts all $n-1$ bundles,
and hence $|T|\le n-2$. Apply Lemma~\ref{lemma:pm_outside_S*}
to this graph and match every agent in $S\setminus T$ to a
distinct bundle outside $L(T)$. There are $n-1-|T|$ matched
bundles. Let $K\subseteq[n-1]$ be their index set, so that each
agent in $T$ assigns cost greater than $1$ to every {$B_h$ with
$h\in K$.}

Allocate these matched bundles, return $D$, and let $R$ be the
set of all remaining chores. Equivalently,
\[
R=(M\setminus P_j)\setminus\bigcup_{h\in K}B_h.
\]
Note that $R$
includes all of $D$. For every $i\in T$,
\[
\begin{aligned}
c_i(R)
&=c_i(M\setminus P_j)-\sum_{h\in K}c_i(B_h)\\
&<(n-1)-(n-1-|T|)=|T|.
\end{aligned}
\]
By Lemma~\ref{lem:prop}, we now obtain a PROP1 allocation of $R$
to the agents in $T$ with at least one PROP bundle. As before,
we know that every proportional threshold is below $1$. By removing at most
one chore from each of the other bundles, we give every remaining
agent her original MMS guarantee. We use at {most
$|T|-1\le n-3$ disposals.}

Only chores left unallocated at the end count as disposals.
In the perfect-matching case, only \(D\) is discarded.
Otherwise, all chores in \(D\) are returned before round-robin.
The resulting allocation requires at most \(|T|-1\) disposals
in total, including any returned chores from \(D\) that are
subsequently discarded. Thus, we discard at most
\(|D|\le n-3\) chores in the first case and \(|T|-1\le n-3\) in the
second.
Figure~\ref{fig:thm41-near-packing} illustrates {the two outcomes.}
\end{proof}

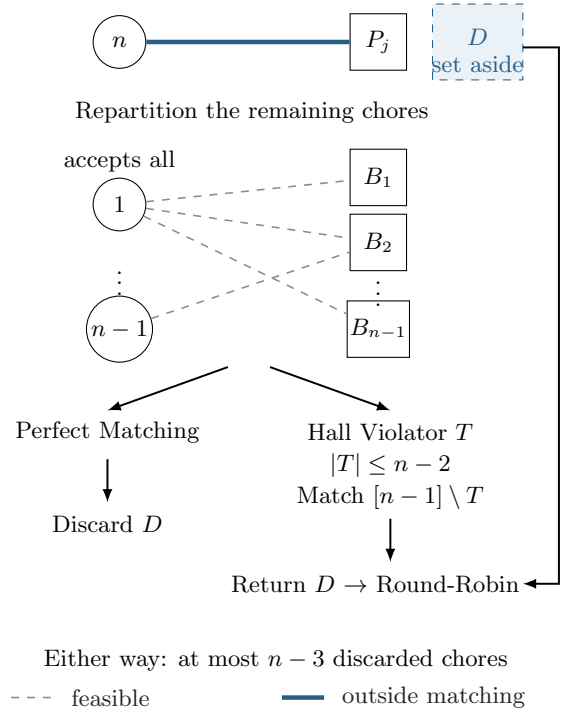
\begin{figure}[!htbp]
    \centering
    \definecolor{mmsInk}{RGB}{32,36,42}
    \definecolor{mmsblue}{RGB}{40,95,134}
    \definecolor{mmspale}{RGB}{234,241,247}

    \resizebox{0.50\linewidth}{!}{%
    \begin{tikzpicture}[
        x=1cm, y=.90cm,
        font=\fontsize{9}{9.4}\selectfont,
        >=latex,
        agent/.style={
            circle, draw, fill=white,
            minimum size=7mm, inner sep=1pt
        },
        bundle/.style={
            rectangle, draw, fill=white,
            minimum size=7.5mm, inner sep=1pt
        },
        feasible/.style={
            dashed, draw=black!45, line width=.6pt
        },
        matched/.style={
            draw=mmsblue, line width=1.7pt
        },
        flow/.style={
            ->, line width=.75pt
        }
    ]
        \path[use as bounding box]
            (0,-1.28) rectangle (8,9.44);

        \node[agent] (an) at (1.92,8.72) {$n$};
        \node[bundle] (Pj) at (5.36,8.72) {$P_j$};
        \draw[matched] (an) -- (Pj);

        \draw[dashed, draw=mmsblue, fill=mmspale]
            (6.08,8.16) rectangle (7.28,9.28);

        \node[text=mmsblue]
            at (6.68,8.80) {$D$};

        \node[
            text=mmsblue,
            font=\fontsize{9}{8.2}\selectfont
        ] at (6.68,8.40) {set aside};

        \node at (3.68,7.68)
            {Repartition the remaining chores};

        \node[agent] (a1) at (1.92,6.32) {$1$};
        \node[agent] (aLast) at (1.92,4.48) {$n-1$};
        \node at (1.92,5.28) {$\vdots$};

        \node[bundle] (B1) at (5.36,6.72) {$B_1$};
        \node[bundle] (B2) at (5.36,5.76) {$B_2$};
        \node[bundle] (BLast) at (5.36,4.48) {$B_{n-1}$};
        \node at (5.36,5.12) {$\vdots$};

        \foreach \b in {B1,B2,BLast} {
            \draw[feasible] (a1) -- (\b);
        }
        \draw[feasible] (aLast) -- (B2);

        \node[font=\fontsize{9}{8.2}\selectfont]
            at (1.92,6.96) {accepts all};

        \draw[flow] (3.36,3.92) -- (1.76,3.28);
        \draw[flow] (3.92,3.92) -- (5.52,3.28);

        \node at (1.76,2.96) {Perfect Matching};
        \node at (5.52,2.96) {Hall Violator $T$};
        \node at (5.52,2.48) {$|T|\leq n-2$};

        \draw[flow] (1.76,2.56) -- (1.76,1.92);
        \node at (1.76,1.60) {Discard $D$};

        \node at (5.52,2) {Match $[n-1]\setminus T$};

        \draw[flow] (5.52,1.68) -- (5.52,1.04);
        \node at (5.32,.72)
            {Return $D$ $\rightarrow$ Round-Robin};

        \draw[flow]
            (7.28,8.64) -- (7.76,8.64)
            -- (7.76,.72) -- (7.28,.72);

        \node at (4,-.40)
            {Either way: at most $n-3$ discarded chores};

        \draw[feasible] (.48,-.96) -- (1.04,-.96);
        \node[
            anchor=west, text=mmsInk,
            font=\fontsize{9}{8.26}\selectfont,
            inner sep=0pt, outer sep=0pt
        ] at (1.28,-.96) {feasible};

        \draw[matched] (4.08,-.96) -- (4.64,-.96);
        \node[
            anchor=west, text=mmsInk,
            font=\fontsize{9}{8.26}\selectfont,
            inner sep=0pt, outer sep=0pt
        ] at (4.88,-.96) {outside matching};

    \end{tikzpicture}%
    }

    \caption{Repartition and Match}
    \label{fig:thm41-near-packing}
\end{figure}

\section{Existence of \texorpdfstring{$(n-1,n-3)$}{(n-1,n-3)}-Packable Set}
\label{sec:n-3 proof}

\paragraph{Overview.}
By Lemma~\ref{lem:near-packing-reduction}, we need only find
\(j\in J\) such that \(M\setminus P_j\) is
\emph{\((n-1,n-3)\)-packable} to agent~\(1\).
We examine what failure of this property would imply.
Although we cannot establish \emph{\((n-1)\)-packability} from
\(c_1(M\setminus P_j)<n-1\) alone, we remove the most costly
chores to obtain a stronger bound in terms of the largest remaining
cost. If the remaining set is still not \emph{\((n-1)\)-packable},
we use this bound to find a subset in which every bundle of cost
at most~\(1\) satisfies a common cardinality restriction.
We then compare the subsets arising from two different complements
with agent~\(1\)'s original MMS partition to obtain a
contradiction. We first establish the structural property underlying this comparison.

\medskip

Fix $P\in\{P_j\}_{j\in J}$, and suppose that
$M\setminus P$ is not \emph{$(n-1,n-3)$-packable} to agent~$1$.
Let $D$ consist of the $n-3$ most costly chores in $M\setminus P$
under $c_1$, and set $M'=(M\setminus P)\setminus D$.
By our assumption, $M'$ is still not \emph{$(n-1)$-packable}.
Furthermore, every chore in
$D$ costs at least as much as every chore in $M'$.
Together with $c_1(P)>1$ and $c_1(M)\le n$, we obtain
\[
\begin{aligned}
c_1(M')
&=c_1(M\setminus P)-c_1(D)\\
&<n-1-(n-3)\cdot \max_{e\in M'}c_1(e).
\end{aligned}
\]

Next, we describe the structure forced by this inequality
for a general number $k$ of bundles. We will apply {it with
$k=n-1$.}

\begin{lemma}\label{lem:critical-core}
Let $M'$ be a set of chores with nonnegative additive costs $c$
and $c(e)\le 1$ for every $e\in M'$, such that
\[
c(M')<k-(k-2)\cdot \max_{e\in M'}c(e)
\]
for some integer $k\ge 3$.
If $M'$ is not \emph{$k$-packable}, then there exists a subset
{$M^*\subseteq M'$ such that}
\begin{itemize}
    \item $|M^*|=kr+1$ for {some integer $r\ge 1$;}
    \item every subset $B\subseteq M^*$ of size {$r+1$ satisfies
    $c(B)>1$.}
\end{itemize}
\end{lemma}
\begin{proof}
Choose a minimum-size subset $M^*\subseteq M'$ that is not
\emph{$k$-packable}, and let
\[
a=\min_{e\in M^*}c(e)
\quad \text{ and } \quad
b=\max_{e\in M^*}c(e).
\]
By nonnegativity and the assumed total-cost bound, we have
\begin{equation*}
    c(M^*)+(k-2)\cdot b < k.
\end{equation*}
We also have $a>0$: otherwise, removing a zero-cost chore would
leave a smaller set that admits a $k$-packing, and adding that
chore back would preserve feasibility.
Choose $e^*\in M^*$ with $c(e^*)=a$.
By minimality, $M^*\setminus\{e^*\}$ admits a {$k$-packing
$\{B_1,\ldots,B_k\}$.}

We first show that each bundle has enough slack to allow any
two chores to be swapped.
No bundle can accept $e^*$, since this would give a $k$-packing
of $M^*$. Thus we have $1-c(B_h)<a$ for every $h\in[k]$.
On the other hand, we have
\[
\sum_{h=1}^k\bigl(1-c(B_h)\bigr)
=k-c(M^*)+a
>(k-2)\cdot b+a.
\]
For any fixed $h$, the other $k-1$ bundles have total slack
less than $(k-1)\cdot a$. Subtracting this upper bound, we obtain
\[
\begin{aligned}
1-c(B_h)
&>(k-2)\cdot b+a-(k-1)\cdot a\\
&=(k-2)\cdot (b-a)\ge b-a,
\end{aligned}
\]
where we use $k\ge3$ in {the last inequality.}

A swap increases the cost of either affected bundle by at most
$b-a$, so we can perform any swap while preserving
feasibility. We obtain the same bound after the swap:
the resulting partition still packs $M^*\setminus\{e^*\}$,
its total slack is unchanged, and no bundle can accept $e^*$
because $M^*$ is not \emph{$k$-packable}.
We can therefore repeat the preceding calculation for the new partition.
Inductively, we can perform any sequence of swaps while preserving feasibility.
We may consequently rearrange the chores arbitrarily while
keeping the cardinality of {each bundle fixed.}

Using this freedom to rearrange chores, we show that all bundles have
equal cardinality.
Suppose otherwise, and relabel a smallest and a largest bundle
as $B_1$ and $B_k$, with sizes $q$ and $s$, respectively,
where $s\ge q+1$.
Use swaps to place the $q$ cheapest chores of
$M^*\setminus\{e^*\}$ in $B_1$ and the $s$ most costly chores
in $B_k$. These choices can be made disjoint, since the original
partition contains both bundles.
Every chore in $B_1\cup\{e^*\}$ is no more costly than any
chore in $B_k$. As $q+1\le s$, we obtain by nonnegativity
$
c(B_1)+a\le c(B_k)\le1.
$
Thus $B_1$ can accept $e^*$, a contradiction.
All bundles have the same size $r$, and hence $|M^*|=kr+1$.
Moreover, $r\ge1$, since an empty bundle could accept $e^*$,
whose cost is {at most~$1$.}

To prove the second statement, we use swaps to place the $r$
cheapest chores of $M^*\setminus\{e^*\}$ in $B_1$.
The resulting packing is feasible, but $B_1$ still cannot accept
$e^*$, so $c(B_1)+a>1$.
The chores in $B_1\cup\{e^*\}$ are the $r+1$ cheapest chores
of $M^*$, breaking ties consistently with the choice of $e^*$.
Every $(r+1)$-subset of $M^*$, whether or not it contains $e^*$,
has at least this cost and therefore costs more than~$1$.
Figure~\ref{fig:thm41-critical-core} {illustrates this structure.}
\end{proof}
We obtain the counting restriction needed below:
every bundle of cost at most~\(1\) contains at most \(r\) chores
from \(M^*\). Applying the lemma with \(c=c_1\), we bound the number
of such chores in each bundle of agent~\(1\)'s original MMS
partition.

We now use this restriction to complete the proof of
Theorem~\ref{thm:n-3}. If the complements of two rejected bundles
both fail to admit the required packing, we use the lemma to identify
sets of chores with large total cost. We show that these sets
force two bundles of the original MMS partition to have combined
cost greater than~\(2\), contradicting their individual cost
bounds of~\(1\). Thus, at least one complement admits the required
packing, and we apply Lemma~\ref{lem:near-packing-reduction} to finish
the proof.

\begin{figure}[!htbp]
  \centering
  \definecolor{mmsfigblue}{RGB}{36,94,120}

  \resizebox{0.50\linewidth}{!}{%
  \begin{tikzpicture}[
    x=.08cm,y=.072cm,line join=round,
    every node/.style={
      anchor=center,text=mmsInk,
      font=\fontsize{8.5}{10.03}\selectfont,
      inner sep=0pt,outer sep=0pt
    }
  ]
    \definecolor{mmsInk}{RGB}{32,36,42}
    \definecolor{mmsBlue}{RGB}{40,95,134}
    \definecolor{mmsGray}{RGB}{139,146,153}
    \definecolor{mmsPale}{RGB}{234,241,247}
    \definecolor{mmsLight}{RGB}{241,242,243}
    \definecolor{mmsFaint}{RGB}{171,178,185}

    \path[use as bounding box]
      (0,0) rectangle (100,80);

    \path[fill=white,draw=mmsInk,line width=.65pt]
      (19,55.2) rectangle (62,64.8);

    \path[fill=white,draw=mmsInk,line width=.65pt]
      (19,38.2) rectangle (62,47.8);

    \path[fill=white,draw=mmsInk,line width=.65pt]
      (19,18.2) rectangle (62,27.8);

    \draw[
      draw=mmsBlue,line width=.75pt,
      ->,>=latex,shorten <=1pt,shorten >=1pt
    ] (28,54) to[bend left=37.7] (28,49);

    \draw[
      draw=mmsBlue,line width=.75pt,
      ->,>=latex,shorten <=1pt,shorten >=1pt
    ] (36,49) to[bend left=37.7] (36,54);

    \draw[
      draw=mmsGray,line width=.75pt,
      ->,>=latex,shorten <=1pt,shorten >=1pt
    ] (80,59) -- (68,52);

    \path[fill=mmsLight,draw=mmsInk,line width=.7pt]
      (78.5,59.5) rectangle (85.5,66.5);

    \path[fill=mmsPale,draw=mmsBlue,line width=.7pt]
      (23.65,57.65) rectangle (28.35,62.35);

    \path[fill=mmsPale,draw=mmsBlue,line width=.7pt]
      (33.65,57.65) rectangle (38.35,62.35);

    \path[fill=mmsPale,draw=mmsBlue,line width=.7pt]
      (51.65,57.65) rectangle (56.35,62.35);

    \path[fill=mmsPale,draw=mmsBlue,line width=.7pt]
      (23.65,40.65) rectangle (28.35,45.35);

    \path[fill=mmsPale,draw=mmsBlue,line width=.7pt]
      (33.65,40.65) rectangle (38.35,45.35);

    \path[fill=mmsPale,draw=mmsBlue,line width=.7pt]
      (51.65,40.65) rectangle (56.35,45.35);

    \path[fill=mmsPale,draw=mmsBlue,line width=.7pt]
      (23.65,20.65) rectangle (28.35,25.35);

    \path[fill=mmsPale,draw=mmsBlue,line width=.7pt]
      (33.65,20.65) rectangle (38.35,25.35);

    \path[fill=mmsPale,draw=mmsBlue,line width=.7pt]
      (51.65,20.65) rectangle (56.35,25.35);

    \node at (50,75) {Smallest nonpackable set $M^*$};
    \node at (82,63) {$e^*$};

    \node[font=\fontsize{9}{8.26}\selectfont]
      at (82,70) {cheapest};

    \node[font=\fontsize{9}{10.62}\selectfont]
      at (10,60) {$B_1$};

    \node[font=\fontsize{9}{9.44}\selectfont]
      at (45,60) {$\cdots$};

    \node[anchor=west,font=\fontsize{9}{10.62}\selectfont]
      at (66,60) {$r$};

    \node[font=\fontsize{9}{10.62}\selectfont]
      at (10,43) {$B_2$};

    \node[font=\fontsize{9}{9.44}\selectfont]
      at (45,43) {$\cdots$};

    \node[anchor=west,font=\fontsize{9}{10.62}\selectfont]
      at (66,43) {$r$};

    \node[font=\fontsize{9}{10.62}\selectfont]
      at (10,23) {$B_k$};

    \node[font=\fontsize{9}{9.44}\selectfont]
      at (45,23) {$\cdots$};

    \node[anchor=west,font=\fontsize{9}{10.62}\selectfont]
      at (66,23) {$r$};

    \path[fill=mmsInk,draw=none]
      (40,31.8) circle[radius=.02cm];
    \path[fill=mmsInk,draw=none]
      (40,33) circle[radius=.02cm];
    \path[fill=mmsInk,draw=none]
      (40,34.2) circle[radius=.02cm];

    \node[text=mmsBlue,font=\fontsize{9}{8.26}\selectfont]
      at (53,51.5) {swap};

    \node[font=\fontsize{9}{8.26}\selectfont]
      at (83,49) {fits};

    \node[font=\fontsize{9}{8.26}\selectfont]
      at (83,44) {nowhere};

    \node[font=\fontsize{9}{10.62}\selectfont]
      at (50,12) {$|M^*|=kr+1$};

    \node[font=\fontsize{9}{9.44}\selectfont]
      at (50,5) {Any $r+1$ chores in $M^*$ cost more than $1$};

  \end{tikzpicture}%
  }
  \caption{Structure of a nonpackable set.}
  \label{fig:thm41-critical-core}
\end{figure}
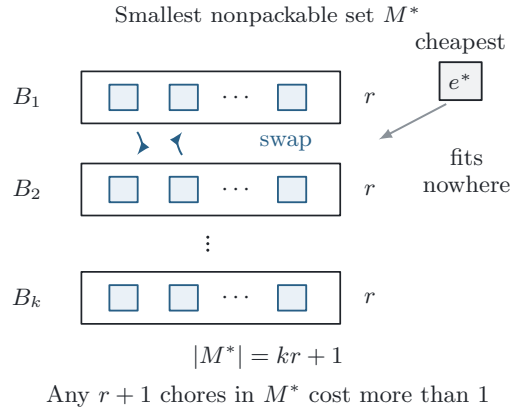

\noindent\textbf{Proof of Theorem~\ref{thm:n-3}.}
By Lemma~\ref{lem:near-packing-reduction}, we need only find an index
\(j\in J\) such that \(M\setminus P_j\) is
\emph{\((n-1,n-3)\)-packable} to agent \(1\).
Suppose, for contradiction, that every such complement fails this condition.
Since \(|J|\ge2\), choose two distinct bundles
\(P,Q\in\{P_j\}_{j\in J}\). These bundles are disjoint and satisfy
{\(c_1(P)>1\) and \(c_1(Q)>1\).}

For each \(Z\in\{P,Q\}\), let \(D_Z\) consist of the \(n-3\)
most costly chores in \(M\setminus Z\) under \(c_1\), breaking ties
arbitrarily. By assumption, the remaining
set \((M\setminus Z)\setminus D_Z\) is not
\emph{\((n-1)\)-packable}.
Each chore in \(D_Z\) costs at least as much as every remaining
chore. Since \(c_1(M)\le n\) and \(c_1(Z)>1\), we obtain
\[
\begin{aligned}
c_1\bigl((M\setminus Z)\setminus D_Z\bigr)
&=c_1(M)-c_1(Z)-c_1(D_Z)\\
&<n-1-(n-3)\cdot
\max_{e\in(M\setminus Z)\setminus D_Z}c_1(e).
\end{aligned}
\]
Every remaining chore satisfies \(c_1(e)\le1\), and
\(n-1\ge3\). Thus we can apply Lemma~\ref{lem:critical-core}
with \(k=n-1\) to obtain a subset
\(M_Z^*\subseteq(M\setminus Z)\setminus D_Z\) and an integer
\(r_Z\ge1\) such that
$
|M_Z^*|=(n-1)\cdot r_Z+1,
$
and every \(r_Z+1\) chores in \(M_Z^*\) have total cost {greater
than \(1\).}

We enlarge the obstruction by restoring the removed chores:
$
\hat M_Z=M_Z^*\cup D_Z.
$
Since \(M_Z^*\) and \(D_Z\) are disjoint, we have
\[
|\hat M_Z|=(n-1)\cdot r_Z+1+(n-3)
=(n-1)\cdot (r_Z+1)-1.
\]
Every chore in \(D_Z\) is at least as costly as every chore in
\(M_Z^*\). Hence, every bundle with \(r_Z+1\) chores in \(\hat M_Z\) still
have cost greater than \(1\). Restoring \(D_Z\) therefore
enlarges the obstruction while preserving the restriction that
each bundle of cost at most \(1\) contains at most \(r_Z\) {of
its chores.}

We now compare these enlarged obstructions with agent~\(1\)'s
original MMS partition. 
Let \((A_1,\ldots,A_n)\) be an MMS partition of agent~\(1\),
so that \(c_1(A_h)\le1\) {for every \(h\in[n]\).} 

\begin{figure}[t]
\centering
\definecolor{packingInk}{HTML}{23313E}
\definecolor{packingMuted}{HTML}{687785}
\definecolor{packingBlue}{HTML}{276A9D}
\definecolor{packingOrange}{HTML}{C07426}
\definecolor{packingRed}{HTML}{AC3541}
\resizebox{0.50\linewidth}{!}{%
\begin{tikzpicture}[
  x=.52cm,y=.5cm,
  every node/.style={font=\fontsize{7.5}{9}\selectfont,text=packingInk,align=center},
  brace/.style={decorate,decoration={brace,amplitude=2.5pt},line width=.5pt}
]
\path[use as bounding box] (-1.20,-1.42) rectangle (15.95,9.80);

\fill[packingBlue!8] (0,0) rectangle (9.5,8.5);
\fill[packingOrange!8] (4.5,0) rectangle (14,8.5);
\fill[packingBlue!7!packingOrange!8] (4.5,0) rectangle (9.5,8.5);

\fill[white,opacity=.60] (0,0) rectangle (14,2.3);
\fill[white,opacity=.60] (0,5.7) rectangle (14,8.5);
\draw[draw=packingBlue!75,line width=.4pt,dashed] (9.5,0) -- (9.5,8.5);
\draw[draw=packingOrange!75,line width=.4pt,dashed] (4.5,0) -- (4.5,8.5);

\fill[packingBlue,opacity=.035] (4.9,4.1) ellipse (4.05 and 3.3);
\fill[packingOrange,opacity=.04] (9.5,4.0) ellipse (4.15 and 3.25);
\draw[draw=packingBlue,line width=.55pt,densely dashed] (4.9,4.1) ellipse (4.05 and 3.3);
\draw[draw=packingOrange,line width=.55pt,densely dashed] (9.5,4.0) ellipse (4.15 and 3.25);

\path[draw=packingBlue,fill=packingBlue!20,line width=.65pt]
 (1.25,4.55) .. controls (1.60,5.37) and (4.40,5.36) .. (5.80,4.98)
 .. controls (7.04,4.60) and (7.18,3.56) .. (6.30,3.10)
 .. controls (5.26,2.66) and (2.69,2.76) .. (1.73,3.32)
 .. controls (1.16,3.68) and (1.02,4.13) .. cycle;
\path[draw=packingOrange,fill=packingOrange!22,line width=.65pt]
 (7.64,4.56) .. controls (8.12,5.40) and (10.79,5.40) .. (12.13,5.08)
 .. controls (13.26,4.77) and (13.36,3.70) .. (12.37,3.10)
 .. controls (11.24,2.55) and (8.62,2.64) .. (7.87,3.24)
 .. controls (7.37,3.61) and (7.24,4.00) .. cycle;

\draw[draw=packingInk!75,line width=.5pt] (0,5.7) -- (14,5.7);
\draw[draw=packingInk!65,line width=.4pt] (0,4) -- (14,4);
\draw[draw=packingInk!75,line width=.5pt] (0,2.3) -- (14,2.3);
\draw[draw=packingInk,line width=.6pt] (0,0) rectangle (14,8.5);
\node[anchor=east,font=\fontsize{8.5}{10}\selectfont] at (-.18,8.40) {$M$};

\draw[brace,draw=packingBlue] (0,8.78) -- (9.5,8.78);
\node[text=packingBlue,font=\fontsize{8.5}{10}\selectfont] at (4.75,9.48) {$M\setminus P$};
\draw[brace,draw=packingOrange] (14,-.28) -- (4.5,-.28);
\node[text=packingOrange,font=\fontsize{8.5}{10}\selectfont] at (9.25,-.98) {$M\setminus Q$};
\node[font=\fontsize{8.5}{10}\selectfont] at (2.25,8.02) {$Q$};
\node[font=\fontsize{7.5}{9}\selectfont,inner sep=.5pt] at (7.00,8.02) {$M\setminus(P\cup Q)$};
\node[font=\fontsize{8.5}{10}\selectfont] at (11.75,8.02) {$P$};
\node[text=packingBlue,font=\fontsize{8.5}{10}\selectfont] at (3.6,6.49) {$\hat{M}_P$};
\node[text=packingOrange,font=\fontsize{8.5}{10}\selectfont] at (10.8,6.49) {$\hat{M}_Q$};

\node[text=packingBlue,font=\fontsize{8.5}{10}\selectfont] at (4.05,4.63) {$T_P$};
\node[text=packingBlue] at (4.05,3.41) {$c_1(T_P)>1$};
\node[text=packingOrange,font=\fontsize{8.5}{10}\selectfont] at (10.24,4.63) {$T_Q$};
\node[text=packingOrange] at (10.24,3.41) {$c_1(T_Q)>1$};

\draw[brace,draw=packingInk!75] (14.24,8.5) -- (14.24,0);
\node[font=\fontsize{8.5}{10}\selectfont] at (15.15,7.32) {$A_1$};
\node[text=packingMuted] at (15.15,6.38) {$\vdots$};
\node[font=\fontsize{8.5}{10}\selectfont] at (15.15,4.85) {$A_i$};
\node[font=\fontsize{8.5}{10}\selectfont] at (15.15,3.15) {$A_{i+1}$};
\node[text=packingMuted] at (15.15,1.98) {$\vdots$};
\node[font=\fontsize{8.5}{10}\selectfont] at (15.15,.68) {$A_n$};

\end{tikzpicture}%
}
\caption{Contradiction in the final proof.}
\label{fig:packing-contradiction}
\end{figure}
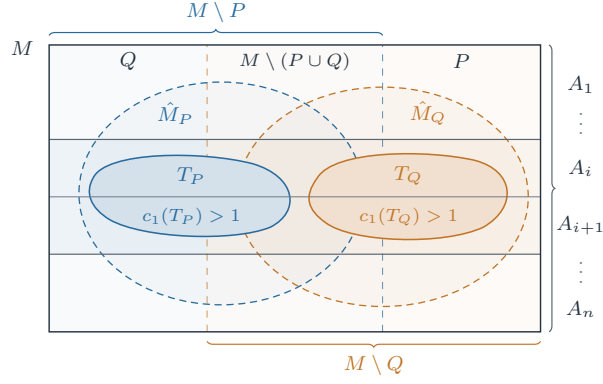

\paragraph{High-level idea.}
For each \(Z\in\{P,Q\}\), the set \(\hat M_Z\) lies outside
the rejected bundle \(Z\) and contains many large chores,
in the sense that any \(r_Z+1\) of them have total cost
greater than \(1\). Let \(\hat M=\hat M_P\cup\hat M_Q\)
collect these large chores. Agent~\(1\)'s original MMS
partition packs all chores into
\(n\) bundles of cost at most \(1\), so each bundle can
contain only a limited number of these chores.
We use a pigeonhole argument to derive a contradiction (see Figure~\ref{fig:packing-contradiction} {for an illustration}):
\begin{itemize}
    \item We first show that \(\hat M\) contains many
    distinct large chores by bounding the overlap between
    \(\hat M_P\) and {\(\hat M_Q\)
    (Lemma~\ref{lem:large-union}}).

    \item We then use the original MMS partition to find two bundles ($A_i$ and $A_{i+1}$ in Figure~\ref{fig:packing-contradiction}) whose union \(H\) contains sufficiently many chores from {\(\hat M\) (Lemma~\ref{lem:dense-mms-bundles}}).

    \item Finally, we turn this counting bound into a cost contradiction by finding two disjoint collections of large chores in $H$, each costing more than $1$. We then obtain \(c_1(H)>2\), contradicting the
    individual cost bounds of the two MMS {bundles
    (Lemma~\ref{lem:disjoint-overload}}).
\end{itemize}

We first show that \(\hat M\) contains many distinct large
chores by using the cost bound to limit the overlap between
{\(\hat M_P\) and \(\hat M_Q\).}

\begin{lemma}\label{lem:large-union}
{We have \(|\hat M|\ge\frac n2(r_P+r_Q+2)-1\).}
\end{lemma}
\begin{proof}
Because \(\hat M_P\cap P=\varnothing\) and
\(\hat M_Q\cap Q=\varnothing\), their intersection is contained
in \(M\setminus(P\cup Q)\). Since \(P\) and \(Q\) are disjoint, we have
\[
c_1\bigl(M\setminus(P\cup Q)\bigr)
=c_1(M)-c_1(P)-c_1(Q)<n-2.
\]
Fix \(Z\in\{P,Q\}\). If the intersection $\hat M_P\cap\hat M_Q$ contained at least
\((n-2)(r_Z+1)\) chores, we could choose that many distinct
chores and divide them into \(n-2\) disjoint groups of size
\(r_Z+1\). Each group is a subset of \(\hat M_Z\) and hence
costs more than \(1\). Their combined cost would exceed \(n-2\),
contradicting the displayed bound. The intersection thus has fewer
than \((n-2)(r_Z+1)\) chores, i.e.,
\[
|\hat M_P\cap\hat M_Q|\le(n-2)(r_Z+1)-1.
\]
Averaging the bounds for the two choices of \(Z\), we obtain
\[
|\hat M_P\cap\hat M_Q|
\le\frac{n-2}{2}(r_P+r_Q+2)-1.
\]
We therefore obtain
\begin{align*}
|\hat M|
&=|\hat M_P|+|\hat M_Q|-|\hat M_P\cap\hat M_Q|\\
&\ge(n-1)(r_P+r_Q+2)-2-\left(\frac{n-2}{2}(r_P+r_Q+2)-1\right)\\
&=\frac n2(r_P+r_Q+2)-1.
\qedhere
\end{align*}
\end{proof}

We next show that some pair of MMS bundles contains many large {chores from \(\hat M\).}
\begin{lemma}\label{lem:dense-mms-bundles}
There exist two bundles among \(A_1,\ldots,A_n\) whose union
{\(H\) satisfies \(|\hat M\cap H|\ge r_P+r_Q+2\).}
\end{lemma}
\begin{proof}
Choose the two bundles among \(A_1,\ldots,A_n\) containing the most chores from \(\hat M\),
and let \(H\) be their union. The average of these two counts
is at least the average over all \(n\) bundles. The bundles
partition \(M\), so their counts sum to \(|\hat M|\). Hence we have
\[
|\hat M\cap H|\ge\frac{2|\hat M|}{n}.
\]
Rounding up and applying Lemma~\ref{lem:large-union} with
\(n\ge4\), we obtain
\begin{equation*}
|\hat M\cap H|
\ge\left\lceil\frac{2|\hat M|}{n}\right\rceil
\ge\left\lceil r_P+r_Q+2-\frac2n\right\rceil
=r_P+r_Q+2.
\qedhere
\end{equation*}
\end{proof}

Fix such a union \(H\). The next lemma extracts two disjoint
groups of large chores, one from each of {\(\hat M_P\)
and \(\hat M_Q\).}
\begin{lemma}\label{lem:disjoint-overload}
There exist disjoint subsets \(T_P\subseteq\hat M_P\cap H\)
and \(T_Q\subseteq\hat M_Q\cap H\) such that
\(|T_P|=r_P+1\) and \(|T_Q|=r_Q+1\).
Consequently, {we have \(c_1(H)>2\).}
\end{lemma}
\begin{proof}
For each \(Z\in\{P,Q\}\), an MMS bundle contains at most
\(r_Z\) chores from \(\hat M_Z\): otherwise, it would contain
an \((r_Z+1)\)-subset of cost greater than \(1\).
The \(n-2\) bundles outside \(H\) therefore contain at most
\((n-2)\cdot r_Z\) such chores. Excluding their contribution, we obtain
\[
\begin{aligned}
|\hat M_Z\cap H|
&\ge(n-1)\cdot (r_Z+1)-1-(n-2)\cdot r_Z\\
&=r_Z+n-2\ge r_Z+1.
\end{aligned}
\]
Each set therefore supplies enough chores for its required selection.
To ensure disjointness, choose
\(T_P\subseteq\hat M_P\cap H\) of size \(r_P+1\), using as
many chores outside {\(\hat M_Q\) as possible.}
\begin{itemize}
    \item If at least \(r_P+1\) chores of \(\hat M_P\cap H\) lie outside
\(\hat M_Q\), then \(T_P\) is disjoint from \(\hat M_Q\).
Using the individual bound above, we can choose any \(r_Q+1\)
chores of {\(\hat M_Q\cap H\) as \(T_Q\).}

    \item Otherwise, \(T_P\) includes every chore of
\((\hat M_P\cap H)\setminus\hat M_Q\), together with enough
chores from the intersection to reach size \(r_P+1\).
Every chore of \((\hat M\cap H)\setminus T_P\) then belongs
to \(\hat M_Q\). By Lemma~\ref{lem:dense-mms-bundles}, we have
\[
\bigl|(\hat M_Q\cap H)\setminus T_P\bigr|
=|\hat M\cap H|-(r_P+1)\ge r_Q+1.
\]
Thus we can choose \(T_Q\) of the required size, disjoint from
\(T_P\). In both cases, we use the defining properties of
\(\hat M_P\) and \(\hat M_Q\), together with nonnegative
costs, to obtain
\[
c_1(H)\ge c_1(T_P)+c_1(T_Q)>2.
\]
\end{itemize}

Figure~\ref{fig:mms-three-set-selection} illustrates {the two cases.}
\end{proof}

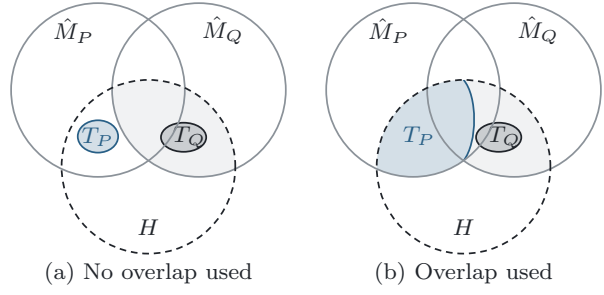
\begin{figure}[t]
\centering
\resizebox{0.50\linewidth}{!}{%
\begin{tikzpicture}[
  x=.53cm,
  y=.53cm,
  line join=round,
  every node/.style={
    anchor=center,
    text=mmsInk,
    font=\fontsize{8.5}{10.03}\selectfont,
    inner sep=0pt,
    outer sep=0pt
  }
]
\definecolor{mmsInk}{RGB}{32,36,42}
\definecolor{mmsBlue}{RGB}{40,95,134}
\definecolor{mmsGray}{RGB}{139,146,153}
\definecolor{mmsLight}{RGB}{241,242,243}

\path[use as bounding box]
  (-3.55,-4.18) rectangle (11.35,3.0);

\foreach \offset in {0,7.8} {
  \begin{scope}[shift={(\offset,0)}]
    \begin{scope}
      \clip (0,-1.2) circle[radius=2.15];
      \fill[mmsLight]
        (1.25,.7) circle[radius=2.15];
    \end{scope}
  \end{scope}
}

\filldraw[
  fill=mmsBlue!20,
  draw=mmsBlue,
  line width=.65pt
]
  (-1.25,-.45) ellipse[x radius=.50,y radius=.40];

\node[text=mmsBlue] at (-1.25,-.45) {$T_P$};

\begin{scope}[shift={(7.8,0)}]
  \begin{scope}
    \clip (0,-1.2) circle[radius=2.15];
    \clip (-1.25,.7) circle[radius=2.15];

    \fill[mmsBlue!20]
      (-4,-4) -- (-4,4) -- (0,4) -- (0,.95)
      .. controls (.35,.55) and (.35,-.60)
      .. (0,-1.049285)
      -- (0,-4) -- cycle;

    \draw[mmsBlue,line width=.65pt]
      (0,.95)
      .. controls (.35,.55) and (.35,-.60)
      .. (0,-1.049285);
  \end{scope}

  \node[text=mmsBlue] at (-1.10,-.42) {$T_P$};
\end{scope}

\foreach \offset in {0,7.8} {
  \begin{scope}[shift={(\offset,0)}]
    \filldraw[
      fill=mmsGray!45,
      draw=mmsInk,
      line width=.65pt
    ]
      (.87,-.47) ellipse[x radius=.56,y radius=.36];

    \node at (1.02,-.47) {$T_Q$};


  \end{scope}
}

\foreach \offset in {0,7.8} {
  \begin{scope}[shift={(\offset,0)}]
    \draw[mmsGray,line width=.65pt]
      (-1.25,.7) circle[radius=2.15];

    \draw[mmsGray,line width=.65pt]
      (1.25,.7) circle[radius=2.15];

    \draw[
      mmsInk,
      line width=.65pt,
      dash pattern=on 2.5pt off 1.8pt
    ]
      (0,-1.2) circle[radius=2.15];

    \node at (-1.85,2.18) {$\hat M_P$};
    \node at (1.85,2.18) {$\hat M_Q$};
    \node at (0,-2.7) {$H$};
  \end{scope}
}

\node at (0,-3.84) {(a) No overlap used};
\node at (7.8,-3.84) {(b) Overlap used};

\end{tikzpicture}%
}
\caption{Choose disjoint \(T_P\) and \(T_Q\).}
\label{fig:mms-three-set-selection}
\end{figure}

Since \(H\) is the union of two bundles from agent \(1\)'s
MMS partition, we also have \(c_1(H)\le2\), which is a contradiction.
Therefore,
at least one of \(M\setminus P\) and \(M\setminus Q\) is
\emph{\((n-1,n-3)\)-packable} to agent \(1\).
Applying Lemma~\ref{lem:near-packing-reduction}, we obtain an allocation
that satisfies every agent's original MMS guarantee and leaves at
most \(n-3\) chores unallocated, completing the proof {of
Theorem~\ref{thm:n-3}.}

\section{Conclusion and Future Directions}

We studied MMS allocations for indivisible chores with nonnegative
additive costs, allowing a limited number of chores to remain
unassigned. We proved that at most \(n-3\) disposals suffice for
\(n\ge4\), independently of the number of chores, while preserving
every agent's original MMS guarantee. For four agents, we establish
the {optimal worst-case bound.}

The main open question is how far the disposal bound can be reduced.
A first step is to improve \(n-3\) to \(n-4\) for \(n\ge5\), and more
generally to \(n-k\) for larger constants \(k\) and sufficiently
large \(n\). A further direction is to obtain smaller linear bounds,
such as \(n/2\) or \(n/4\), and ultimately a sublinear bound \(o(n)\).
The strongest goal is to prove that, for any number of agents with
nonnegative additive chore costs, discarding at most one chore
always suffices to guarantee every agent {her original MMS.}

\section*{Declaration for the Use of AI}
We used OpenAI's GPT-5.6 Sol, GPT-6 Astra and GPT-6.1 Sol to assist with mathematical analysis, language editing, notation consistency checks, and LaTeX formatting. All mathematical statements, proofs, calculations, and citations were reviewed and verified by the authors, who take full responsibility for the content of the paper.

\clearpage
\bibliographystyle{alpha}
\bibliography{ref}

\clearpage
\appendix
\section{MMS with \texorpdfstring{$n-3$}{n-3} Goods Duplications}
\label{gdup:sec:main}

We prove the goods analogue of Theorem~\ref{thm:n-3} using the same
matching and repartitioning framework. For $n\ge4$, we improve the
$n-2$ bound of Akrami et al.~\cite[Lemma~3.8]{akrami2026augmentation}
to $n-3$, with each good copied at most once and its two instances
assigned to different agents.

\begin{theorem}
\label{gdup:thm:main}
For every instance $I=(N,M,\{v_i\}_{i\in N})$ with $|N|=n\ge4$
and nonnegative additive valuations, there exists a complete
allocation with at most $n-3$ goods duplications that gives every
agent her original MMS guarantee. Each original good is copied
at most once, and no agent receives both the original and its copy.
\end{theorem}

\subsection{Preliminaries}

We use the notation of the chores setting, with nonnegative
additive valuations $v_i$ for goods.

\begin{definition}[MMS]
\label{gdup:def:mms}
Agent $i$'s maximin share is
$$\MMS_i:=\max_{\bX\in\prod_n(M)}\min_{j\in[n]}v_i(X_j).$$
A bundle is MMS-feasible to $i$ if its value is at least $\MMS_i$.
An MMS partition of $i$ is an $n$-partition of $M$ into such bundles.
\end{definition}

We may assume that, for every agent $i\in N$,
\begin{equation}
\label{gdup:eq:normalization}
\MMS_i=1,\qquad 0\le v_i(e)\le1\quad(e\in M).
\end{equation}
If all MMS values are zero, any complete allocation suffices.
Otherwise, an agent with positive MMS has $n$ nonempty bundles
in her MMS partition, so $|M|\ge n$. Give each zero-MMS agent
unit value for every good; her original guarantee is automatic.
Scale each valuation to $\MMS_i=1$, then replace each singleton
value $v_i(e)$ by $\min\{v_i(e),1\}$ and extend additively.
Every bundle worth at least $1$ remains so, and capping cannot
increase MMS; hence MMS remains $1$. An MMS partition therefore
gives $v_i(M)\ge n$.
Any allocation meeting these normalized guarantees, including
with copies, also meets the original guarantees. Henceforth,
$v_i$ denotes the normalized valuation.

\begin{definition}[MMS Allocation]
\label{gdup:def:mms-allocation}
A complete allocation $\bX$ of $M$ is MMS if $v_i(X_i)\ge1$ for
every $i\in N$.
\end{definition}

\begin{definition}[MMS Allocation with Duplications]
\label{gdup:def:duplication-allocation}
An MMS allocation with at most $\ell$ duplications assigns all
original goods and at most $\ell$ labeled copies, giving every agent
her original MMS guarantee. A copy has the same value as its original
to every agent. Each original is copied at most once, and its two
instances go to different agents.
\end{definition}

These duplication restrictions correspond to the distinct-copy case of
Akrami et al.~\cite[Definition~2.5]{akrami2026augmentation}.
Throughout, augmented partitions obey these restrictions.

\begin{definition}[MMS-feasibility Graph]
\label{gdup:def:feasibility-graph}
For an $n$-partition $\bP=(P_1,\ldots,P_n)$, the graph is defined
as in Definition~\ref{chores-MMS-feasibility graph}, with an edge $(i,P_j)$ whenever
$v_i(P_j)\ge1$. We again write $L(S)$ for the neighborhood of an
agent set $S$.
\end{definition}

\subsection{Allocation Framework}
\label{gdup:sec:reduction}

As in Section~\ref{four agents}, fix an MMS partition $\bP$ of agent $n$ and form
its feasibility graph. Every agent has a neighbor by
$v_i(M)\ge n$, and agent $n$ accepts every bundle. A perfect matching
gives a complete MMS allocation. Otherwise, choose a
maximum-cardinality Hall violator $S$, meaning $|L(S)|<|S|$.
Then $2\le|S|\le n-1$.

As shown in Lemma~\ref{lemma:pm_outside_S*}, we have the following observation for goods.
\begin{lemma}
\label{gdup:lem:outside}
There exists a matching saturating $N\setminus S$ using only
bundles in $\bP\setminus L(S)$.
\end{lemma}

Similar to Lemma~\ref{lem:prop}, our allocation steps use the following guarantee, while keeping copies distinct.
\begin{lemma}[{\cite[Lemma~3.7]{akrami2026augmentation}}]
\label{gdup:lem:bagfill}
Suppose $t\ge1$ agents divide a set of original goods $R$, with
$v_i(R)>t$ and $v_i(e)\le1$ for every participating agent $i$
and $e\in R$. There is a complete allocation giving each agent
value at least $1$, using at most $t-1$ copies of distinct originals
in $R$, with no agent receiving both instances of a good.
\end{lemma}

To obtain this formulation, apply the cited lemma to
$u_i=t\,v_i/v_i(R)$, for which $u_i(R)=t$, $u_i(e)<1$, and
$u_i\le v_i$; assign any unused originals without adding copies.

\begin{lemma}
\label{gdup:lem:small}
If $|S|\le n-2$, there exists an MMS allocation with at most
$n-3$ duplications.
\end{lemma}

\begin{proof}
As in Lemma~\ref{lem:non-extreme-hall}, match the agents outside $S$ to bundles rejected
by all agents in $S$. Here the remaining goods $R$ satisfy
$v_i(R)>|S|$ for every $i\in S$. By
Lemma~\ref{gdup:lem:bagfill}, these agents can be satisfied using
at most $|S|-1\le n-3$ copies of distinct originals in $R$.
\end{proof}

\begin{remark}
The same argument gives an MMS allocation with at most $n-2$
distinct duplications whenever $n\ge2$, recovering the bound of
Akrami et al.~\cite[Lemma~3.8]{akrami2026augmentation}.
\end{remark}

It remains to consider $S=N\setminus\{n\}$. Let
$J=\{j\in[n]:P_j\notin L(S)\}$, so $|J|\ge2$ and
$v_i(P_j)<1$ for all $i\in S$ and $j\in J$. Following the
reduction in Lemma~\ref{lem:near-packing-reduction}, we seek a
complement that becomes $n-1$ MMS-feasible bundles for one remaining
agent after adding at most $n-3$ distinct copies. We also require
the following forest condition.

\begin{definition}[$(k,\ell)$-Coverable]
\label{gdup:def:coverable}
For integers $k\ge1$ and $\ell\ge0$, a set of labeled goods is
\emph{$k$-coverable} to agent $i$ if it admits a partition into
$k$ bundles of value at least $1$ to $i$.
A set of originals $M'\subseteq M$ is \emph{$(k,\ell)$-coverable}
to $i$ if adding at most $\ell$ copies of distinct originals in $M'$
yields such a partition, with no bundle containing both instances
of a good, and such that the following \emph{duplication multigraph}
$D$ is a forest: its vertices are the bundles, and each duplicated
good contributes one edge joining the bundles containing its
original and its copy. Thus $D$ has no cycles or parallel edges.
We call $D$ the \emph{duplication forest} and the partition a
\emph{forest partition}.
\end{definition}

For a set of original goods, nonnegativity makes $k$-coverability
equivalent to the existence of $k$ pairwise disjoint bundles of
value at least $1$.
The forest condition prevents copy counts from accumulating in
the second matching step.

\begin{lemma}
\label{gdup:lem:reduction}
If $M\setminus P_j$ is $(n-1,n-3)$-coverable to some $i\in S$
for some $j\in J$, then there exists an MMS allocation with at
most $n-3$ duplications.
\end{lemma}

\begin{proof}
Relabel the agents in $S$ so that the chosen agent is $1$.
Let $(B_1,\ldots,B_{n-1})$ be a forest partition of
$M\setminus P_j$ for agent $1$, with at most $n-3$ copies,
and let $D$ be its duplication forest. Allocate $P_j$ to agent
$n$, leaving $v_i(M\setminus P_j)>n-1$ for every $i\in S$.
As in Lemma~\ref{lem:near-packing-reduction}, form a new graph
on $S$ and the actual augmented bundles, with an edge $(i,B_h)$
whenever $v_i(B_h)\ge1$. Agent $1$ accepts every bundle, and
every other agent has a neighbor since
$\sum_h v_i(B_h)\ge v_i(M\setminus P_j)>n-1$.
A perfect matching therefore gives the required allocation.

Otherwise, let $T$ be a maximum-cardinality Hall violator.
It excludes agent $1$, so $t:=|T|\le n-2$.
By Lemma~\ref{gdup:lem:outside}, match $S\setminus T$ to
distinct bundles outside $L(T)$. Let $K$ be their indices and
write $a=|K|=n-1-t\ge1$.
For each duplicated good with one instance in the bundles indexed
by $K$ and the other outside these bundles, label the former as
the original and the latter as the copy. Allocate the bundles indexed by $K$ and cancel all copies
in the other bundles. The retained copies correspond to edges of
$D[K]$, so their number is $d=|E(D[K])|\le a-1$.

Let $U_K$ be the set of distinct originals represented in these
allocated bundles, and let $R=(M\setminus P_j)\setminus U_K$.
For every $i\in T$, nonadjacency gives
$v_i(U_K)\le\sum_{h\in K}v_i(B_h)<a$ and hence
$v_i(R)>n-1-a=t$.
By Lemma~\ref{gdup:lem:bagfill}, we can satisfy all agents in $T$
using at most $t-1$ fresh copies of distinct originals in $R$.
The total is at most $d+t-1\le(a-1)+(t-1)=n-3$ copies.
Every original is allocated, and all duplication restrictions hold.
\end{proof}

\subsection{Existence of \texorpdfstring{{$(n-1,n-3)$-Coverable}}{(n-1,n-3)-Coverable} Set}
\label{gdup:sec:proof}
We now show that some complement is \emph{$(n-1,n-3)$-coverable}
to agent~$1$. We prove a covering counterpart of
Lemma~\ref{lem:critical-core}, first for three covers and then
for a general number of covers while preserving a duplication
forest. For a nonnegative additive valuation $v$, a bundle $B$
is a \emph{cover} if $v(B)\ge1$; it is \emph{minimal} if no
proper subset is a cover.

\begin{claim}
\label{gdup:clm:path-cover}
If $v(M')>s$ for an integer $s\ge1$ and $v(e)\le1$ for every
$e\in M'$, then $M'$ is {$(s,s-1)$-coverable}.
\end{claim}

\begin{proof}
Represent the positive-valued goods by consecutive intervals
whose lengths equal their values. Divide the resulting interval
into $s$ equal parts, each of length $\frac{v(M')}{s}>1$, and
give each bundle one instance of every good whose interval
overlaps its part in positive length. Since each good has value
at most $1$, it crosses at most one cut, and each cut is crossed
by at most one good. Thus every bundle has value greater than
$1$, at most $s-1$ distinct goods are copied once each, and the
duplication graph is a subgraph of a path. Assign zero-valued
goods arbitrarily.
\end{proof}

\begin{lemma}
\label{gdup:lem:structure}
Let $M'$ be a finite set of original goods with nonnegative
additive values $v(e)\le1$ and $v(M')>k$, where $k\ge3$ is an
integer. If $M'$ is not {$(k,k-2)$-coverable}, then there
exist an integer $r\ge2$ and a set $M^*\subseteq M'$ of size
$k(r-1)+1$ such that, for every $Y\subseteq M'$,
$$
v(Y)<1
\quad\Longleftrightarrow\quad
|Y\cap M^*|\le r-1.
$$
Equivalently, the minimal covers are precisely the $r$-subsets
of $M^*$.
\end{lemma}

\begin{proof}
If some original has value $1$, allocate it alone and apply
Claim~\ref{gdup:clm:path-cover} to the remaining goods, whose
total value exceeds $k-1$. This gives a forest partition with at
most $k-2$ copies, a contradiction. Hence every original has
value below $1$.

We first establish an exchange property for three covers.
Let $W\subseteq M'$ satisfy $v(W)>3$, and suppose that $W$
is not \emph{$(3,1)$-coverable}. We show that any two minimal
covers in $W$ have the same cardinality, and every subset of
their union of that cardinality is a minimal cover. 

First take two disjoint minimal covers $B_1,B_2\subseteq W$,
and let $L=W\setminus(B_1\cup B_2)$. We have $v(L)<1$;
otherwise, these sets would give three disjoint covers.
Every $x\in B_1\cup B_2$ satisfies $v(x)<1-v(L)$, since
otherwise $L+\operatorname{copy}(x)$ would be a third cover.
Moreover, minimality gives $v(x)>v(B_h)-1$ for every $x\in B_h$, $h=1,2$.

Swap any $x\in B_1$ and $y\in B_2$. Both resulting bundles
remain covers, since
$$
\begin{aligned}
v(B_1-x+y)
&>v(B_1)+v(B_2)+v(L)-2\\
&=v(W)-2>1,
\end{aligned}
$$
and similarly for $B_2-y+x$. They also remain minimal.
Otherwise, some good $z$ could be removed from one of the new
covers while leaving a cover. Choose a good $w$ from the
pre-swap cover that did not contain $z$. Then $z\ne w$, and
minimality of the pre-swap covers gives
$$
\begin{aligned}
v\bigl(L+z+\operatorname{copy}(w)\bigr)
&=v(L)+v(z)+v(w)\\
&>v(L)+v(B_1)+v(B_2)-2\\
&=v(W)-2>1.
\end{aligned}
$$
After removing $z$, the original $w$ remains in one of the
new covers. Thus $L+z+\operatorname{copy}(w)$ is a legal third
cover, a contradiction.

Swaps leave $B_1\cup B_2$ and $L$ unchanged and preserve
minimality, so they can be repeated. The two covers must have
the same cardinality. Otherwise, use swaps to place the
least-valued goods in the smaller cover and the remaining
goods in the larger one. Removing any good from the larger
cover would then leave a bundle worth at least as much as
the smaller cover, contradicting minimality. Every subset
of their union with this common cardinality can be obtained
by swaps and is therefore a minimal cover.

Now take arbitrary minimal covers $B,C\subseteq W$.
Since $v(B)<2$ and $v(W)>3$, there is a minimal cover
$D\subseteq W\setminus B$. The disjoint-cover argument gives
$|B|=|D|=:r$, where $r\ge2$, and shows that every $r$-subset
of $U:=B\cup D$ is a minimal cover. If $C\subseteq U$,
minimality forces $|C|=r$, and every $r$-subset of $B\cup C$
is already minimal. Otherwise, $|C\cap U|<r$, since an
$r$-subset of this intersection would be a proper covering
subset of $C$. Hence $|U\setminus C|\ge r+1$.
For any $r$-subset $Y\subseteq B\cup C$, choose an $r$-subset
$E\subseteq U\setminus C$ containing $Y\setminus C$.
Applying the disjoint-cover argument to $C,E$ gives $|C|=r$
and shows that $Y$ is a minimal cover. This proves the
three-cover exchange property.

We next show that any two intersecting minimal covers in $M'$
have the same cardinality, and every subset of their union
of that cardinality is a minimal cover. For $k=3$, this follows
from the exchange property with $W=M'$. Assume $k\ge4$, take
intersecting minimal covers $B,C$, and choose $z\in B\cap C$.
Minimality gives $v(B),v(C)<1+v(z)$, so
$$
v(B\cup C)
\le v(B)+v(C)-v(z)
<2+v(z)<3.
$$
Hence the goods outside $B\cup C$ have total value greater
than $k-3$.

Put $t=k-3$. Apply the bag-filling construction of
Akrami et al.~\cite[Algorithm~2 and proof of Lemma~3.7]{akrami2026augmentation}
to the goods outside $B\cup C$, starting with a maximum-valued
original $f$ and stopping after $t$ allocations, with an instance
of the last duplicated good retained for continuation.
Label the instance retained in the remaining set $R$ as the original $g$. The construction gives $t$ covers using $t$ copies
of distinct originals and leaves a set $R$ of originals
containing $B\cup C$. By its value-loss bound,
$
v(R)>v(M')-t>3.
$
The only previously copied original in $R$ is $g$, while $f$
lies in the first allocated bundle, has not been copied, and
satisfies $v(f)\ge v(g)$. The first $t$ bundles form a path
in the duplication graph, and the edge for $g$ will join its
last vertex to the bundle receiving the original $g$.

Suppose that $R$ is \emph{$(3,1)$-coverable}, then it has a forest partition with three bundles and at
most one copy. If it uses no copy or copies an original other
than $g$, combine the allocations directly. Since $g$ is the
only previously copied original in $R$, no original is copied
more than once. The path on the first $t$ bundles attaches to
one residual bundle through $g$, and there is at most one edge
between residual bundles. The resulting graph is a forest
with at most $t+1=k-2$ edges.

If the residual partition copies $g$, replace that copy by a
fresh copy of $f$. This does not decrease its recipient's value.
The original $f$ has not been copied and lies in the first
allocated bundle, outside $R$. The residual bundles containing
the original $g$ and its former copy are distinct. The former
is joined to the last of the first $t$ bundles by the edge for
$g$, and the latter to the first by the new edge for $f$.
The resulting graph is a path with the third residual bundle
isolated, using exactly $t+1=k-2$ copies of distinct originals.
Either case contradicts the hypothesis that $M'$ is not \emph{$(k,k-2)$-coverable}.

Thus $R$ is not \emph{$(3,1)$-coverable}. Applying the
three-cover exchange property to $W=R$ shows that $B,C$
have the same cardinality and every subset of their union
of that cardinality is minimal. Since minimality depends
only on the bundle, this property also holds in $M'$.

For any two disjoint minimal covers $B,C\subseteq M'$, there
is a minimal cover intersecting both. Indeed, choose a
maximum-valued good $x\in B\cup C$, relabeling so that $x\in B$.
For any $y\in C$, the set $C-y+x$ is a cover. A minimal cover
$D\subseteq C-y+x$ must contain $x$, since $C-y$ is not a
cover, and at least one good from $C$, since $v(x)<1$.
Hence $D$ intersects both $B$ and $C$. The exchange property
for intersecting covers therefore implies that all minimal
covers have the same size $r\ge2$.

Starting with any minimal cover, add goods from $M'$ while
preserving the property that every $r$-subset is a minimal
cover, until no further good can be added. Denote the resulting
set by $M^*$; then $|M^*|\ge r$. Suppose a minimal
cover $B$ intersects $M^*$ but is not contained in it.
For any $r$-subset $Y\subseteq M^*\cup B$, if
$Y\subseteq M^*$, then $Y$ is already minimal. Otherwise,
$|Y\setminus B|\le r-1$. Choose an $r$-subset
$C\subseteq M^*$ containing $Y\setminus B$ and an element
of $B\cap M^*$. Then $B,C$ intersect and $Y\subseteq B\cup C$,
so the exchange property makes $Y$ minimal. This contradicts
the maximality of $M^*$.

Every minimal cover is contained in $M^*$. Indeed, fix a
minimal cover $C\subseteq M^*$ and take any minimal cover $B$.
If $B$ intersects $C$, the preceding argument gives
$B\subseteq M^*$. Otherwise, there is a minimal cover $D$
intersecting both $B$ and $C$. Applying the same argument
first to $D$ and then to $B$ again gives $B\subseteq M^*$.
Thus the minimal covers are precisely the $r$-subsets of
$M^*$, which gives the stated equivalence.

If $|M^*|\le k(r-1)$, partition all originals into $k$ sets,
each containing at most $r-1$ goods from $M^*$. Every set
has value below $1$, contradicting $v(M')>k$. Hence
$|M^*|\ge k(r-1)+1$.

Suppose that $|M^*|\ge k(r-1)+2$. If $|M^*|\ge kr$, form
$k$ disjoint $r$-subsets of $M^*$ and allocate all other
originals arbitrarily, obtaining $k$ covers without copies.
Otherwise, put $d=kr-|M^*|$, so $1\le d\le k-2$.
We now construct $k$ bundles, each containing $r$ distinct goods
from $M^*$, using $d$ copies and ensuring that the duplication
graph is a forest. Take a path with $d$ edges on $d+1$ of the $k$ bundle vertices,
leaving the others isolated. Assign a distinct original from
$M^*$ to each edge and place its two instances in the endpoint
bundles. Each bundle receives at most $2\le r$ goods this way.
Fill all $k$ bundles to $r$ goods each using distinct,
previously unused originals from $M^*$. The $d$ duplicated goods occupy $2d$ of the $kr$ places.
The remaining $kr-2d$ places are filled with distinct unused
originals, so the total number of distinct originals used is
$$
d+(kr-2d)=kr-d=|M^*|.
$$ Each bundle contains $r$ distinct goods
from $M^*$ and is therefore a cover. Exactly $d\le k-2$
originals are copied once each, and the duplication graph
is the chosen forest. Assign all remaining originals
arbitrarily; none has been copied. In either case, $M'$
is \emph{$(k,k-2)-$coverable}, a contradiction. Therefore,
$|M^*|=k(r-1)+1$.

\end{proof}

\paragraph{Proof of Theorem~\ref{gdup:thm:main}.}
Suppose for a
contradiction that every complement $M\setminus P_j$,
$j\in J$, fails to be \emph{$(n-1,n-3)$-coverable} to agent~$1$.
Choose distinct $P,Q\in\{P_j:j\in J\}$. For each $Z\in\{P,Q\}$,
we have $v_1(M\setminus Z)\ge n-v_1(Z)>n-1$.
Lemma~\ref{gdup:lem:structure}, with $k=n-1$, gives a set
$\hat M_Z\subseteq M\setminus Z$ and an integer $r_Z\ge2$
such that $|\hat M_Z|=(n-1)(r_Z-1)+1$ and
\begin{equation}
\label{gdup:eq:core-characterization}
v_1(Y)<1
\quad\Longleftrightarrow\quad
|Y\cap\hat M_Z|\le r_Z-1
\qquad(Y\subseteq M\setminus Z).
\end{equation}
Let $\hat M=\hat M_P\cup\hat M_Q$, and fix an MMS partition
$(A_1,\ldots,A_n)$ of agent~$1$. Since $\MMS_1=1$, we have
$v_1(A_h)\ge1$ for every $h\in[n]$.

As in Lemmas~\ref{lem:large-union}--\ref{lem:disjoint-overload},
we bound $|\hat M|$ and select two MMS bundles containing
few goods from $\hat M$ to derive a contradiction.

\begin{lemma}
\label{gdup:lem:small-union}
We have $r_P=r_Q=:r$, $\hat M_P\setminus Q=\hat M_Q\setminus P$,
and $|\hat M|\le n(r-1)+1$.
\end{lemma}

\begin{proof}
Let $X=M\setminus(P\cup Q)$. Since
$v_1(X)=v_1(M)-v_1(P)-v_1(Q)>n-2\ge2$, it contains a
minimal cover. By \eqref{gdup:eq:core-characterization},
the minimal covers in $X$ are precisely the $r_P$-subsets
of $X\cap\hat M_P$, and also the $r_Q$-subsets of
$X\cap\hat M_Q$. A minimal cover in $X$ therefore gives
$r_P=r_Q=:r$. Both $X\cap\hat M_P$ and $X\cap\hat M_Q$ consist exactly of
the goods appearing in minimal covers in $X$. Thus
$$
X\cap\hat M_P
=
\bigcup_{\substack{B\subseteq X\\ B\text{ is a minimal cover}}} B
=
X\cap\hat M_Q.
$$ Since
$\hat M_P\subseteq M\setminus P$ and
$\hat M_Q\subseteq M\setminus Q$, it follows that
\begin{equation}
\label{gdup:eq:core-decomposition}
\hat M_P\setminus Q=\hat M_Q\setminus P.
\end{equation}
Also, $v_1(P)<1$ implies $|P\cap\hat M_Q|\le r-1$.
Therefore,
$$
\begin{aligned}
|\hat M|
&=|\hat M_P|+|P\cap\hat M_Q|\\
&\le(n-1)(r-1)+1+(r-1)
=n(r-1)+1.
\end{aligned}
$$
\end{proof}

\begin{lemma}
\label{gdup:lem:sparse-mms-bundles}
There exist two bundles among $A_1,\ldots,A_n$ whose union
$H$ satisfies $|\hat M\cap H|\le2(r-1)$.
\end{lemma}

\begin{proof}
Choose the two bundles containing the fewest goods from
$\hat M$, and let $H$ be their union. Averaging gives
$$
|\hat M\cap H|
\le\left\lfloor\frac{2|\hat M|}{n}\right\rfloor
\le\left\lfloor2(r-1)+\frac{2}{n}\right\rfloor
=2(r-1).
$$
\end{proof}

Fix such a union $H$.

\begin{lemma}
\label{gdup:lem:deficient-parts}
The union $H$ admits a partition into
$Y_P\subseteq M\setminus P$ and $Y_Q\subseteq M\setminus Q$
such that $|Y_P\cap\hat M_P|\le r-1$ and
$|Y_Q\cap\hat M_Q|\le r-1$. Consequently, $v_1(H)<2$.
\end{lemma}

\begin{proof}
Initially set $Y_P=H\cap Q$ and $Y_Q=H\cap P$.
Since $P\cap Q=\varnothing$, these sets are disjoint and
satisfy $Y_P\subseteq M\setminus P$ and
$Y_Q\subseteq M\setminus Q$. Since $Y_P\subseteq Q$ and $Y_Q\subseteq P$,
\eqref{gdup:eq:core-characterization} and
$v_1(P),v_1(Q)<1$ give
$$
|Y_P\cap\hat M_P|\le r-1,
\qquad
|Y_Q\cap\hat M_Q|\le r-1.
$$
By \eqref{gdup:eq:core-decomposition}, every remaining good
of $H\cap\hat M$ belongs to both $\hat M_P$ and $\hat M_Q$.
Distribute the still-unassigned goods of $H\cap\hat M$ one at a time between $Y_P$ and $Y_Q$,
keeping $|Y_P\cap\hat M_P|\le r-1$ and
$|Y_Q\cap\hat M_Q|\le r-1$. This is always possible:
if both intersections already had size $r-1$ while a good
remained unassigned, then $|H\cap\hat M|>2(r-1)$,
a contradiction.
Assign all still unassigned goods of $H$ arbitrarily between
$Y_P$ and $Y_Q$. These goods lie outside $P\cup Q$ and
$\hat M$, so $Y_P\subseteq M\setminus P$ and
$Y_Q\subseteq M\setminus Q$ still hold, and the two bounds
are unchanged. Thus \eqref{gdup:eq:core-characterization} gives
$$
v_1(H)=v_1(Y_P)+v_1(Y_Q)<2.
$$
\end{proof}

But $H$ consists of two MMS bundles, so $v_1(H)\ge2$, a contradiction.
Thus at least one of $M\setminus P$ and $M\setminus Q$
is \emph{$(n-1,n-3)$-coverable} to agent~$1$.
Applying Lemma~\ref{gdup:lem:reduction} gives the desired
allocation, which also satisfies every agent's original
MMS guarantee by the normalization argument. Hence, we finish the proof of Theorem~\ref{gdup:thm:main}.
\end{document}